\documentclass[letterpaper, 10 pt, conference]{ieeeconf}  % Comment this line out if you need a4paper

\IEEEoverridecommandlockouts                              % This command is only needed if 
\usepackage{cite}

\usepackage{amsmath,amssymb,amsfonts}
\usepackage{amsthm}

\usepackage{graphicx}
\usepackage{subcaption}
\usepackage{textcomp}
\usepackage{xcolor,subfig,multirow,caption}
\newtheorem{prop}{Proposition}

\newtheorem{theo}{Theorem}
\newtheorem{coro}{Corollary}

\newtheorem{exmp}{Example}[section]
\usepackage[norelsize, linesnumbered, ruled, lined, boxed, commentsnumbered, noend]{algorithm2e}

\usepackage[normalem]{ulem}
\usepackage{mathabx}

\ifodd 0
\newcommand{\com}[1]{{\color{red}\textbf{Comment}:#1}}
\else

\newcommand{\com}[1]{}
\fi

\newcommand{\comment}[1]{}

\title{\LARGE \bf A Stackelberg Game Model of Flocking 
}

\author{Chenlan Wang$^{1}$, Mehrdad Moharrami$^{2}$, and Mingyan Liu$^{1}$% <-this % stops a space
\thanks{$^{1}$ University of Michigan, Ann Arbor, USA. $^{2}$  University of
Iowa, Iowa City, USA. %This work is supported by the ARO under contract W911NF1810208.
}
}

\begin{document}

\maketitle
\thispagestyle{empty}
\pagestyle{empty}
%%%%%%%%%%% forcing page %%%%%%%%%%
% \thispagestyle{plain}
% \pagestyle{plain}

%%%%%%%%%%%%%%%%%%%%%%%%%%%%%%%%%%%%%%%%%%%%%%%%%%%%%%%%%%%%%%%%%%%%%%%%%%%%%%%%
\begin{abstract}
 We study a Stackelberg game to examine how two agents determine to cooperate while competing with each other. Each selects an arrival time to a destination, the earlier one fetching a higher reward. There is, however, an inherent penalty in arriving too early as well as a risk in traveling alone. This gives rise to the possibility of the agents cooperating by traveling together while competing for the reward. In our prior work \cite{wang2023cooperation} we studied this problem as a sequential game among a set of $N$ competing agents in continuous time, and defined the formation of a group traveling together as arriving at exactly the same time. In the present study, we relax this definition to allow arrival times within a small window, and study a $2$-agent game in both continuous and discrete time, referred to as the flock formation game. We derive and examine the properties of the subgame perfect equilibrium (SPE) of this game. 
\end{abstract}

% \input{intro}
%%%%%%%%%%%%%%%%%%%%%%%%%%%%%%%%%%%%%%%%%%%%%%%%%%%%%%%%%%%%%%%%%%%%%%%%%%%%%%%%
\section{Introduction}
Individuals often form groups for purposes such as enhanced protection,  more efficient resource pooling and consumption, etc. \cite{pitcher1982fish,sandler1980clubs}, even while they continue to compete at the individual level.  This is widely observed in both human societies and the world of wildlife. In the case of migratory birds, for instance, individuals form flocks during migration to reduce individual predation risk \cite{beauchamp2021flocking}, enhance navigation accuracy \cite{beauchamp2011long}, and achieve higher energy efficiency \cite{mirzaeinia2019energy}. 
However, competition for territories often starts immediately upon arrival at the breeding ground since a better territory often leads to higher reproductive success \cite{janiszewski2014timing}.

The formation of groups has been discussed and explored in many fields, including economics, computer science, social science, political science, and biology \cite{couzin2006behavioral,goette2006impact,hollard2000existence,milchtaich2002stability}. Different techniques and methods are introduced to study how stable groups are formed including game theory \cite{banerjee2001core, bogomolnaia2002stability,slikker2001coalition,wang2024structural, kokko1999competition, janiszewski2014timing, wang2023cooperation}, clustering \cite{van2013community,wittemyer2005socioecology}, and agent-based modeling \cite{collins2017agent,collins2018strategic}. In the context of classical game theory, the formation of groups is most typically studied as strategic, one-shot games (such as coalition formation games \cite{banerjee2001core, slikker2001coalition}, hedonic games \cite{bogomolnaia2002stability}, and its extensions \cite{wang2024structural}).

Two notable exceptions that study how groups form {\em in time} using {\em sequential} game models, where individuals make choices asynchronously, are Kokko \cite{kokko1999competition} and our prior work Wang et al. \cite{wang2023cooperation}, both motivated by behaviors of flocking. 
Specifically, $N$ agents each decide on an arrival time (to the breeding ground), knowing that there is a natural/biologically optimal time of arrival (and thus deviation in either direction carries a penalty) and that earlier arrival allows them to claim the best territory/reward of those remaining. The key difference between these two studies is that while \cite{kokko1999competition}
primarily focuses on the competitive relationship among individuals in their decisions, \cite{wang2023cooperation} explicitly introduces a risk term in the agents' utility function to capture the benefit of being part of a group, and shows that this results in a much richer set of subgame-perfect equilibria (SPE).

A key modeling assumption in both \cite{kokko1999competition} and 
\cite{wang2023cooperation} is that agents are only considered to be part of a group (or flock, used interchangeably in this paper) if they arrive at exactly the same time, which is assumed to be continuous. In this paper we seek to relax this key assumption and extend the model presented in \cite{wang2023cooperation} to discrete time as well. 
The motivation is that in many real-world scenarios groups (and the benefit they afford) are often more loosely defined. 
In the same context of migratory birds, predators such as hawks usually hunt their prey at the frequency of hourly or lower in the spring \cite{lindstrom1989finch}. Therefore, for those arriving at the same migratory stopover site within the same time frame of one to several hours,  the risk for each decreases with the total number of individuals arriving within that time frame rather than a single point in time.

We will refer to this relaxed notion of flocking simply as flocking, while referring to the more stringent requirement in \cite{wang2023cooperation} as {\em strict flocking}; the resulting game is referred to as the {\em flock formation} (FF) game. For ease of exposition, we will also limit our model to $N=2$ agents, which results in a two-step or Stackelberg game. The remainder of this paper is organized as follows. Section~\ref{sec:flock formation game}  introduces the Stackelberg game model. Sections \ref{sec:2-agent-continuous} and \ref{sec:discreteWFFG} derive in detail the SPEs in the continuous-time and discrete-time settings, respectively. Section~\ref{sec:discussion} summarizes and compares the results in this work with that obtained using the more restrictive definition of strict flocking \cite{wang2023cooperation}. Section \ref{sec:conclusion} concludes the paper.

\section{Game model and Preliminaries}
\label{sec:flock formation game}

%%%%%%%%%%%%%%%%%%%%%%%%%%%%%%%%%%%%%%%%%%%%%%%%%%%%%%%%%%%%%%%%%%%%%%%%
% \subsection{The model} \label{sec:model}
The basic game model is very similar to that introduced in our prior work \cite{wang2023cooperation}, with the crucial difference of relaxing the definition of a group.  We will also limit our analysis to two agents.

Consider two agents, denoted $\mathcal{N}=\{1,2\}$, traveling to a destination in order to reach and compete for territories/rewards. Traveling together allows them to reduce travel expenses, such as increased energy efficiency and lower (predation) risk (e.g., in the case of migratory birds). On the other hand, they are also interested in arriving early to secure a territory with higher quality (e.g., in terms of food resources, nesting materials, and protection from predators). We are primarily interested in understanding when the two agents will decide to travel together, i.e., collaborate, despite their competitive relationship. 

Each agent has a positive attribute called {\em strength}, denoted by $\beta_i >0$, $i\in\mathcal{N}$, which represents the natural quality of the agent (e.g., surviving skills, foraging ability, flight experience, etc.). 
Similarly, each territory is associated with a positive value representing its quality, denoted by $E_i>0$, $i\in\mathcal{N}$. 
Without loss of generality, we will assume that Agent 1 is stronger than Agent 2: $\beta_1 > \beta_2$, and that territory 1 is better than territory 2: $E_1 > E_2 > 0$. A direct consequence of the strength attribute is that the same journey is less costly for the stronger one than it is for the weaker one, and that if they arrive at the same time, the stronger one will claim the better territory. 

An agent's decision is its time of arrival at the destination, denoted by $t_i$, $i=1,2$, with a joint action profile $\mathbf{t} = (t_1, t_2)$ or $\mathbf{t} = (t_i, {t}_{-i})$, following the notational convention. 

\subsection{Assumptions} 

\noindent (1) Full residency effect: whoever first occupies a territory will get to keep it. In other words, we assume a first-come first-serve model for the agents to claim the territories. 

\noindent (2) Tie-breaking: If both agents arrive at the same time, the stronger agent claims the better territory.  

\noindent (3) Sequential decision making: the stronger, Agent 1, the leader, chooses its arrival time $t_1$, which is announced and followed by Agent 2, the follower, who decides on $t_2$. It is assumed that everyone commits to their decisions, i.e., agent $i$ will indeed arrive at the decided and announced time $t_i$.  

%%%%%%%%%%%%%%%%%%%%%%%%
\subsection{The utility function} 
The utility of agent $i$ is given as follows: 
\begin{equation}\label{eqn:u1} 
    u_i(\mathbf{t}) = e_i(\mathbf{t}) - c_i(t_i) - \bar{p}_i(\mathbf{t})~, 
\end{equation}
where $e_i(\mathbf{t})$ is the reward from the territory that agent $i$ obtains, $c_i(t_i)$ is the (travel) cost dependent alone on the agent's chosen time, and $\bar{p}_i(\mathbf{t})$ is another cost that depends on all agents' chosen times.  This last term will also be referred to as the (predation) risk term. 
We detail each term below. 
\paragraph{Benefit} if $t_i < t_{-i}$, $e_i(\mathbf{t}) = E_1$; if $t_i > t_{-i}$, $e_i(\mathbf{t}) = E_2$; if $t_i = t_{-i}$, based on the tie-breaker rule in the assumptions, $e_i(\mathbf{t}) = E_1$ for Agent 1 and $e_i(\mathbf{t}) = E_2$ for Agent 2.

\paragraph{Travel cost} $c_i(t_i)= \frac{1}{\beta_i}(t_i -t_o)^2 +c^i_o$, where $t_o$ is an {\em optimal arrival time}, at which point the travel cost is minimized, to a fixed $c^i_o$; deviation in either direction will increase the travel cost, and weaker agents (smaller $\beta$) are more sensitive to the sub-optimality of this deviation.  This model captures the cost of travel purely due to external factors (such as climate in the case of spring migration, where traveling during colder or warmer weather can be detrimental).  The fixed cost $c^i_o$ is agent-dependent, and generally lower for a stronger agent. However, the presence of this fixed cost does not impact our subsequent analysis since the agent's decision-making is entirely relative to the optimal $t_o$.  For this reason and without loss of generality, we will set $c^i_o=0$ for the rest of the paper.  
 \paragraph{Predation risk} $\bar{p}_i(\mathbf{t}) = \frac{r}{|\{j:|t_j-t_i| \leq w\}|}$, where $|\{j:|t_j-t_i| \leq w\}|$ is the total number of agents arriving within the $w$ unit of time before or after $t_i$ including agent $i$ itself, and $r$ is the (nominal) risk for a single individual. For simplicity, we will let $w=1$, which can be understood as 1 unit of time (i.e., an hour, a day, or a week). Since this is a 2-agent case, what this risk definition says is that whenever the two agents' arrival times fall within a time window of size $w$, they are considered to be traveling as a flock, and they each benefit from a risk reduction in half. This is the main modeling difference between the present paper and our prior work \cite{wang2023cooperation}, which assumes a strict flocking definition (requiring $w=0$); it can be seen that the present flocking definition is both a relaxation and a generalization.

To understand how agents make their choices in the face of resource competition but also the potential benefit in collaboration, we will examine the subgame perfect equilibrium (SPE) of the flock formation game in two different settings: continuous time and discrete time. The next two sections present results in each case, respectively. 

\subsection{Key results in the Strict Flocking Game (SFG) \cite{wang2023cooperation} for two agents}

There are only two types of SPEs in the 2-agent game when a flock is defined  strictly as those arriving at exactly the same time:
\begin{enumerate}
    \item $(t_o, t_o)$: When the difference between the two territories is small,  agents cooperate and travel together, as the benefit of flocking outweighs the differences in territorial quality. This is called the cooperation SPE.
    \item $( t_0 - \sqrt{(E_1-E_2)\beta_2},t_0)$: When the territorial difference is significant, the advantage of securing the better territory outweighs that of flocking. Consequently, the stronger agent arrives earlier to deter the weaker agent from competing for the better territory. The weaker agent, in response, abandons the competition and arrives at the optimal time $t_o$ with a minimal travel cost.
\end{enumerate} 

\section{The Continuous-time Game}\label{sec:2-agent-continuous}

The definition of flocking is such that if arriving slightly earlier than the stronger agent, the weaker agent not only gets the better territory but also benefits from flocking. 
What this means in a continuous-time setting is that Agent 2, being the follower, would never choose to arrive at exactly the same time as Agent 1, as it could simply move up its arrival by an infinitesimal amount to benefit from flocking without paying more in travel cost and taking the better territory (this is actually a disadvantage faced by the first move in this game).  Note that this is in direct contrast to the results from \cite{wang2023cooperation} outlined above, where strict flocking is often observed in an SPE.  

Our main results for the continuous-time game are summarized in the following theorem.  In short, 
there are three and only three types of pure strategy SPEs, and in all three cases, the weaker agent arrives later than the stronger agent. More specifically, if the territory difference ($E_1-E_2$) is not too large, then Agent 1 arrives just early enough to allow Agent 2 to arrive at exactly the optimal time $t_o$ while benefiting from flocking and settling for the inferior territory, in an attempt to discourage Agent 2 from competing for the superior territory. 
If $E_1-E_2$ is substantial, then Agent 1 has to arrive much earlier to secure the superior territory, and Agent 2 would either also move up its arrival to remain in a flock with Agent 1 if its benefit outweighs the increased travel cost, or otherwise arrive alone at time $t_o$.

\begin{prop}\label{prop:t0}
No agents will arrive later than $t_o$ in an SPE in the 2-agent continuous-time game. 
\end{prop} 

This can be established through a simple contradiction: the second mover would never choose a time later than $t_o$ no matter what the first mover does because by moving up its time by an infinitesimal amount the second mover can improve its utility; for this reason, the first mover would never choose a time later than $t_o$, either.

In what follows we will use the notation $t_{i-}$ to denote an arrival time $t_i-\epsilon$ for some small  $\epsilon>0$. Before presenting the main result of this game, it's worth noting that given Agent 1's decision of $t_1 \leq t_o$, Agent 2's best response has to be one of the following three: 
\begin{itemize}
    \item $t_{1-}$: Agent 2 obtains the better territory and benefits from flocking;  arriving any earlier would strictly decrease its utility because the travel cost and predation risk would both increase; 
    \item $t_1 +1$ (if $t_1 +1 \leq t_0$): Agent 2 obtains the worse territory but benefits from flocking; 
    \item $t_o$: Agent 2 obtains the worse territory; it may or may not benefit from flocking, but it has the lowest travel cost;
\end{itemize}

\begin{theo}\label{thm:SPEcon}
There exist three {and only three} types of pure strategy SPEs in the 2-agent continuous-time game: \\
\textbf{Type 1}: A flock with arrival times $\mathbf{t}^*_1 = (t_{1,1}^*,t_{2,1}^*)$:
\begin{align*}
    (t_{1,1}^*,t_{2,1}^*) = (t_o-\sqrt{\beta_2(E_1-E_2)}, t_o);
\end{align*}
\textbf{Type 2}: No flock with arrival times $\mathbf{t}^*_2 = (t_{1,2}^*,t_{2,2}^*)$:
\begin{align*}
(t_{1,2}^*,t_{2,2}^*) = (t_o - \sqrt{\beta_2(E_1 -E_2+\frac{r}{2})}, t_o);
\end{align*}
\textbf{Type 3}: A flock with arrival times $\mathbf{t}^*_3 = (t_{1,3}^*,t_{2,3}^*)$:
\begin{align*}
(t_{1,3}^*,t_{2,3}^*) = (t_o\! -\! \frac{\beta_2(E_1\! -\!E_2)+1}{2}, t_o \!- \!\frac{\beta_2(E_1\!-\!E_2)-1}{2}).
\end{align*}
The conditions for each of these SPEs are listed as follows. \\
(1) If $E_1 - E_2 \leq \frac{1}{\beta_2} $, then $\mathbf{t}^*_1 $ is {the unique} SPE.\\
(2) If $E_1 - E_2 > \frac{1}{\beta_2} $:
\begin{enumerate}
\item If $(E_1-E_2)\beta_2 < \sqrt{2r\beta_2} +1$:
    \begin{enumerate}
    \item If $4(E_1-E_2)\beta_1 \geq ((E_1-E_2)\beta_2+1)^2$, then $\mathbf{t}^*_3$ is {the unique} pure strategy SPE;
    \item  If $4(E_1-E_2)\beta_1 < ((E_1-E_2)\beta_2+1)^2$, then there does not exist any pure strategy SPE.
    \end{enumerate}
\item If $(E_1-E_2)\beta_2 > \sqrt{2r\beta_2} +1$:
    \begin{enumerate}
    \item If $(E_1-E_2-\frac{r}{2})\beta_1 \geq (E_1-E_2+\frac{r}{2})\beta_2$, then $\mathbf{t}^*_2$ is {the unique} pure strategy SPE;
    \item If $(E_1-E_2-\frac{r}{2})\beta_1 < (E_1-E_2+\frac{r}{2})\beta_2$, then there does not exist any pure strategy SPE.
    \end{enumerate}
\item If $(E_1-E_2)\beta_2 = \sqrt{2r\beta_2} +1$:
    \begin{enumerate}
    \item If $(E_1-E_2-\frac{r}{2})\beta_1 \geq (E_1-E_2+\frac{r}{2})\beta_2$, or $4(E_1-E_2)\beta_1 \geq ((E_1-E_2)\beta_2+1)^2$, then $\mathbf{t}^*_2$ or $\mathbf{t}^*_3$ are the pure strategy SPEs, respectively;
    \item If $(E_1-E_2-\frac{r}{2})\beta_1 < (E_1-E_2+\frac{r}{2})\beta_2$, and $4(E_1-E_2)\beta_1 < ((E_1-E_2)\beta_2+1)^2$, then there does not exist any pure strategy SPE.
    \end{enumerate}
\end{enumerate}
\end{theo}

\begin{proof}[Proof Sketch]
    Notice that a pure strategy exists only if agent 2 withdraws from competing with agent 1; otherwise, agent 2's best response will be $t_{1-}$, and there would be no pure strategy SPE. Hence, in any SPE, agent 1 picks a time $t_1$ at which agent 2 is indifferent between $t_{1-}$ and either $t_1+1$ or $t_0$. We refer to this point as the tipping point.\\
    {\bf Case (1): $E_1 - E_2 \leq \frac{1}{\beta_2}$}. 
    Agent 2 will cease competing once its arrival time reaches the tipping point, where it becomes indifferent between $t_{1-}$ and the time when Agent 2 obtains the worse territory with the lowest traveling cost (which is indeed $t_o$) while still benefiting from the weakly flocking (i.e. $t_o-t_1 \leq 1$). Based on calculations (see the full proof in the Appendix for details), the tipping time is $t_1 = t_o-\sqrt{\beta_2(E_1-E_2)} = t_{1,1}^*$. In other words, if $t_1 = t_o-\sqrt{\beta_2(E_1-E_2)}$ (thus $t_1 + 1 \geq t_0$), Agent 2's best response is $t_o$ rather than $t_{1-}$.
    If Agent 1 ends up with the worst territory, then the best utility it can obtain is at $t_0$ with the lowest travel cost and the benefit from the flock. Calculation shows: $u_1(t_0, t_{0-}) < u_1(t_{1,1}^*, t_0)$. Hence, Agent 1's utility decreases if it ends up with the worse territory. Therefore, Agent 1's best response is $t_{1,1}^*$ while Agent 2's best response after observing Agent 1's action is $t_{2,1}^* = t_o$.
    
\noindent{\bf Case (2): $E_1 - E_2 > \frac{1}{\beta_2}$}. 
    Given the larger territory difference in this scenario, leading to increased competition, all three ($t_{1-}$, $t_1+1$ ($t_1+1<t_0$), and $t_0$) are possible best responses for Agent 2 given Agent 1's arrival time $t_1$.
    Comparing $u_2(t_1, t_1 + 1)$ and $u_2(t_1, t_0)$, the case can be separated into three subcases:
    \begin{enumerate}
        \item $u_2(t_1, t_1 + 1) > u_2(t_1, t_0)$: if Agent 2 ends up withdrawing the competition with Agent 1, its best response is $t_1+1$. By calculation, the tipping point ($t_{1,3}^*$) that leads to this withdrawal is the time when Agent 2 is indifferent between $t_{1-}$ and $t_1+1$. \\
        If Agent 1 ends up with the worst territory, then the best utility it can obtain is at $t_o$. Since $t_2^*(t_o) = t_{o-}$, Agent 1's utility is $u_1(t_o, t_{o-})$. If $u_1(t_{1,3}^*, t_{1,3}^*+1) \geq u_1(t_o, t_{o-})$, then $\mathbf{t}^*_3$ is {the unique} pure strategy SPE; otherwise, there does not exist any pure strategy SPE.
        \item $u_2(t_1, t_1 + 1) < u_2(t_1, t_0)$: if Agent 2 ends up withdrawing the competition with Agent 1, its best response is $t_0$. By calculation, the tipping point ($t_{1,2}^*$) that leads to this withdrawal is the time when Agent 2 is indifferent between $t_{1-}$ and $t_0$. \\
        If Agent 1 ends up with the worst territory, then the best utility it can obtain is at $t_o$. Since $t_2^*(t_o) = t_{o-}$, Agent 1's utility is $u_1(t_o, t_{o-})$. If $u_1(t_{1,2}^*, t_0) \geq u_1(t_o, t_{o-})$, then $\mathbf{t}^*_2$ is {the unique} pure strategy SPE; otherwise, there does not exist any pure strategy SPE.
        \item $u_2(t_1, t_1 + 1) = u_2(t_1, t_0)$: if Agent 2 ends up withdrawing the competition with Agent 1, its best response is either $t_0$ or $t_1 + 1$. Therefore, all the SPEs (i.e. $\mathbf{t}^*_2$, $\mathbf{t}^*_3$) in the two other subcases are possible. If $u_1(t_o, t_{o-})$ is the largest among the three: $u_1(t_o, t_{o-})$, $u_1(t_{1,2}^*, t_0)$, $u_1(t_{1,3}^*, t_{1,3}^*+1)$, then there does not exist any pure strategy SPE.
\end{enumerate}
\vspace{-5mm}
\end{proof}

\begin{coro}
The stronger agent always arrives earlier than the weaker
agent in an SPE of the 2-agent continuous-time game: $t_1^* < t_2^*$.
\end{coro}

We now compare the above results with those obtained in the strict flocking game \cite{wang2023cooperation}. Given the more relaxed definition of flocking, the competition between the two agents is more fierce since the weaker agent can more easily enjoy the benefit of flocking. As a result, the cooperation SPE ($t_o, t_o$) observed in the strict flocking game \cite{wang2023cooperation} is no longer valid. Similarly, while an SPE always exists in the strict flocking game, it may not exist in 
the present game as shown in Theorem \ref{thm:SPEcon}. Another noteworthy observation is that while in the strict flocking game the weaker agent always arrives at $t_o$, in the present game it might arrive earlier than $t_o$.

\section{The Discrete-time Game}
\label{sec:discreteWFFG}

We now restrict the agents' actions to discrete-time: $t_i = t_o - k_i$, $k_i \in \mathcal{Z}$. 
In this case, in order to secure the superior territory, Agent 2 will need to arrive exactly one unit of time ahead of Agent 1. Thus, the cost associated with traveling one unit of time earlier becomes relevant, which may prompt cooperation in addition to competition. Indeed, in this discrete-time case, the strict flock $(t_o, t_o)$  emerges as a possible SPE shown below.  

Another, perhaps more interesting observation that emerges in the discrete-time setting is that it is now possible for the weaker agent to arrive earlier than the stronger agent in an SPE, in contrast to what happens in the strict flocking game \cite{wang2023cooperation} or in the continuous-time model presented in the previous section. This new SPE occurs when the stronger agent decides to arrive at $t_o$, knowing the weaker agent will arrive earlier, at $t_o-1$. 
This turns out to be an interesting example of {\em first-mover disadvantage}: Agent 1 may simply have no good options even though it knows Agent 2 will get ahead no matter what it does first; as a matter of fact, when this happens Agent 2 will enjoy a strictly higher utility than Agent 1. Note that the same phenomenon existed in the continuous-time game presented in the previous section, with the direct consequence of the non-existence of pure strategy SPE (see cases: (2)1)b, (2)2)b, (2)3)b)).

Since the arrival times are in the discrete setting, we will use $\lceil x \rceil$ to denote $\text{ceil}(x)$. 

Similar to the reasoning given for the continuous-time game, we note that given Agent 1's decision of $t_1 \leq t_0$, Agent 2's best response has to be one of the following three choices: $t_1 - 1$, $t_1 +1$, and $t_o$.

\begin{theo}\label{thm:SPEdis}
There exist five types of pure strategy SPEs in the 2-agent discrete-time game:\\
\textbf{Type 1}: $ \mathbf{t}^*_1 = (t_o, t_o)$, a strict flock;\\
\textbf{Type 2}: $\mathbf{t}^*_2 = (t_o-1, t_o)$, a flock; \\
\textbf{Type 3}: $\mathbf{t}^*_3 = (t_o - k^*, t_o - k^*+1)$, a flock;\\
\textbf{Type 4}: $\mathbf{t}^*_4 = (t_o - k^*, t_o)$, no flock;\\
\textbf{Type 5}: $\mathbf{t}^*_5 = (t_o , t_o - 1)$, a flock;\\
where $k^* = \lceil\min(\sqrt{(E_1-E_2+\frac{r}{2})\beta_2}, \frac{(E_1-E_2)\beta_2}{4}+1) \rceil - 1$. The conditions for each of these SPEs are listed as follows. \\
(1) If $E_1 - E_2 \leq \frac{1}{\beta_2} $, then $\mathbf{t}^*_1$ is {the unique} SPE.\\
(2) If $E_1-E_2 \in (\frac{1}{\beta_2}, \frac{4}{\beta_2}]$, then $\mathbf{t}^*_2$ is {the unique} SPE.\\
(3) If $E_1-E_2 > \frac{4}{\beta_2}$:
\begin{enumerate}
    \item If $k^* \leq \sqrt \frac{r}{2}\beta_2+1$,
        \begin{enumerate}
        \item If $k^*<\sqrt{(E_1 - E_2)\beta_1}$, then $\mathbf{t}^*_3$ is {the unique} SPE.
        \item If $k^*>\sqrt{(E_1 - E_2)\beta_1}$, then $\mathbf{t}^*_5$ is {the unique} SPE.
        \item If $k^*=\sqrt{(E_1 - E_2)\beta_1}$ is an integer, then $\mathbf{t}^*_3$ and $\mathbf{t}^*_5$ are the two SPEs. 
        \end{enumerate}
    \item If $k^* > \sqrt \frac{r}{2}\beta_2+1$,
        \begin{enumerate}
        \item If $k^*<\sqrt{(E_1 - E_2-\frac{r}{2})\beta_1}$, then $\mathbf{t}^*_4$ is {the unique} SPE.
        \item If $k^*>\sqrt{(E_1 - E_2-\frac{r}{2})\beta_1}$, then $\mathbf{t}^*_5$ is {the unique} SPE.
        \item If $k^*=\sqrt{(E_1 - E_2-\frac{r}{2})\beta_1}$ is an integer, then $\mathbf{t}^*_4$ and $\mathbf{t}^*_5$ are the two SPEs.
        \end{enumerate}
\end{enumerate}
\end{theo}

\begin{proof}[Proof Sketch]
     {\bf Case (1): $E_1 - E_2 \leq \frac{1}{\beta_2}$}\\
    Since $E_1 - E_2 \leq \frac{1}{\beta_2}$, then $u_2 (t_1, t_1) \geq u_2(t_1, t_1 - 1)$ given any $t_1 \leq t_0$. Therefore, the weaker agent strictly prefers to form a strict flock with the stronger agent rather than competing for better territory in a flock. Anticipating $t_2^*(t_1) \geq t_1$, the stronger agent's best response is then $t_1^* = t_o$ as it reaches its global minimum with the lowest traveling cost, the lowest predation cost, and the best territory. Thus, $\mathbf{t}^*_1$ is the unique SPE.

    {\bf Case (2): $E_1-E_2 \in (\frac{1}{\beta_2}, \frac{4}{\beta_2}]$}\\
    Since $E_1-E_2 > \frac{1}{\beta_2}$, then $u_2(t_o, t_o) < u_2(t_o, t_o-1)$. Thus, $(t_o, t_o)$ is no longer an SPE, and Agent 2 wants to compete with Agent 1 for better territory. Therefore, Agent 1 has to arrive at or earlier than $t_o - 1$; otherwise, it will end up with a worse territory and a lower utility. If $t_1 = t_o -1$, then $u_2(t_1, t_o) \leq u_2(t_1, t_1 - 1)$ and thus Agent 2 is either indifferent between arriving one unit time earlier than arriving later in the flock or strictly prefers the latter (i.e. $t_2^* (t_1 = t_o-1) = t_o$).  Therefore, Agent 1's best response, anticipating Agent 2's decision, is $t_1^* = t_o -1$. As a result,  $\mathbf{t}^*_2$ is the unique SPE.

    {\bf Case (3): $E_1-E_2>\frac{4}{\beta_2}$}\\ 
When Agent 1 arrives at $t_o - k_1$ and obtains the best territory, the best response of Agent 2 is either $t_o - k_1+1$ or $t_o$ (i.e. $t_o - k_1-1$ is no better than at least one of the other two choices). This also implies that Agent 1 can't secure the best territory if it arrives one unit time later at $t_o - k_1 + 1$, which means that Agent 2 is willing to arrive earlier to get the best territory. Therefore, if $t_1 = t_o - k_1 + 1$, then $t_o - k_1$ is the best response for Agent 2, and thus both $t_o - k_1+2$ and $t_o$ are strictly worse than $t_o - k_1$. 

If Agent 1 does not secure the best territory $E_1$ and ends up with the worst territory $E_2$, then its best response is to arrive at $t_o$ with the lowest travel cost. Observing $t_1 = t_o$, Agent 2's best response is $t_o -1$ as $u_2(t_o, t_o-1) > u_2(t_o, t_o)$. 

Depending on the sets of conditions and which of the two cases above occurs, we have three possible SPEs: $\mathbf{t}^*_3$, $\mathbf{t}^*_4$, and $\mathbf{t}^*_5$.
Detailed calculations can be found in the Appendix.
\end{proof}

When competition for the better territory intensifies, i.e., when $E_1 \gg E_2$, the stronger agent prioritizes securing the better territory, even at a high travel cost. We elaborate on this in the following corollary.

\begin{coro}\label{coro:deltaE_large}
    As the territorial difference becomes substantial: $E_1 - E_2 \rightarrow \infty$, $\mathbf{t}^*_4$ becomes the unique SPE in the 2-agent discrete-time game.
\end{coro}

\begin{proof}
   By Theorem \ref{thm:SPEdis}, only $\mathbf{t}^*_3$, $\mathbf{t}^*_4$, $\mathbf{t}^*_5$ are the possible SPEs in Case (3) with $k^* = \lceil\min(\sqrt{(E_1-E_2+\frac{r}{2})\beta_2}, \frac{(E_1-E_2)\beta_2}{4}+1) \rceil - 1$. For all large values of $E_1-E_2$, we have $k^* = \lceil\sqrt{(E_1-E_2+\frac{r}{2})\beta_2} \rceil - 1$. Thus $k^* / \sqrt{(E_1-E_2)} \rightarrow \sqrt{\beta_2}$, and $\sqrt{(E_1-E_2 -\frac{r}{2})\beta_1} / \sqrt{(E_1-E_2)} \rightarrow \sqrt{\beta_1}$. Given $\beta_1 > \beta_2$, $k^* < \sqrt{(E_1-E_2 -\frac{r}{2} )\beta_1} $ for all large values of $E_1-E_2$, leading to Case (3) 2(a), where $\mathbf{t}^*_4$ is the unique SPE.
\end{proof}

We include an example below that illustrates all SPEs from Case (3) in Theorem~\ref{thm:SPEdis}, corresponding to different values of $\Delta_E = E_1 - E_2$.

\begin{exmp}\label{eg:CASE3}
Consider $r=2$, $\beta_1 = 4.5$, and $\beta_2 = 4$. The relations among the four functions in Figure~\ref{fig:DWFFGcase3} determine the type of SPE. Each region corresponds to a unique SPE: $\mathbf{t}_3^*$ - $(1,a_1), (a_2, a_3), (a_4, a_5)$; $\mathbf{t}_4^*$ - $(b_1,b_2), (b_3, b_4), (b_5, +\infty)$; $\mathbf{t}_5^*$ - $(a_1,a_2), (a_3, a_4), (a_5, b_1), (b_2, b_3), (b_4, b_5)$. Note when $\Delta_E > a_5$, it becomes Case (3) 2), where either $\mathbf{t}_4^*$ or $\mathbf{t}_5^*$ is an SPE. When $\Delta_E > b_5$, it becomes Case (3) 2) b), where $\mathbf{t}_4^*$ is the unique SPE.
\end{exmp}

\begin{figure}[!ht]
    \centering
    \begin{subfigure}[b]{0.95\linewidth}
        \centering
        \includegraphics[width=\linewidth]{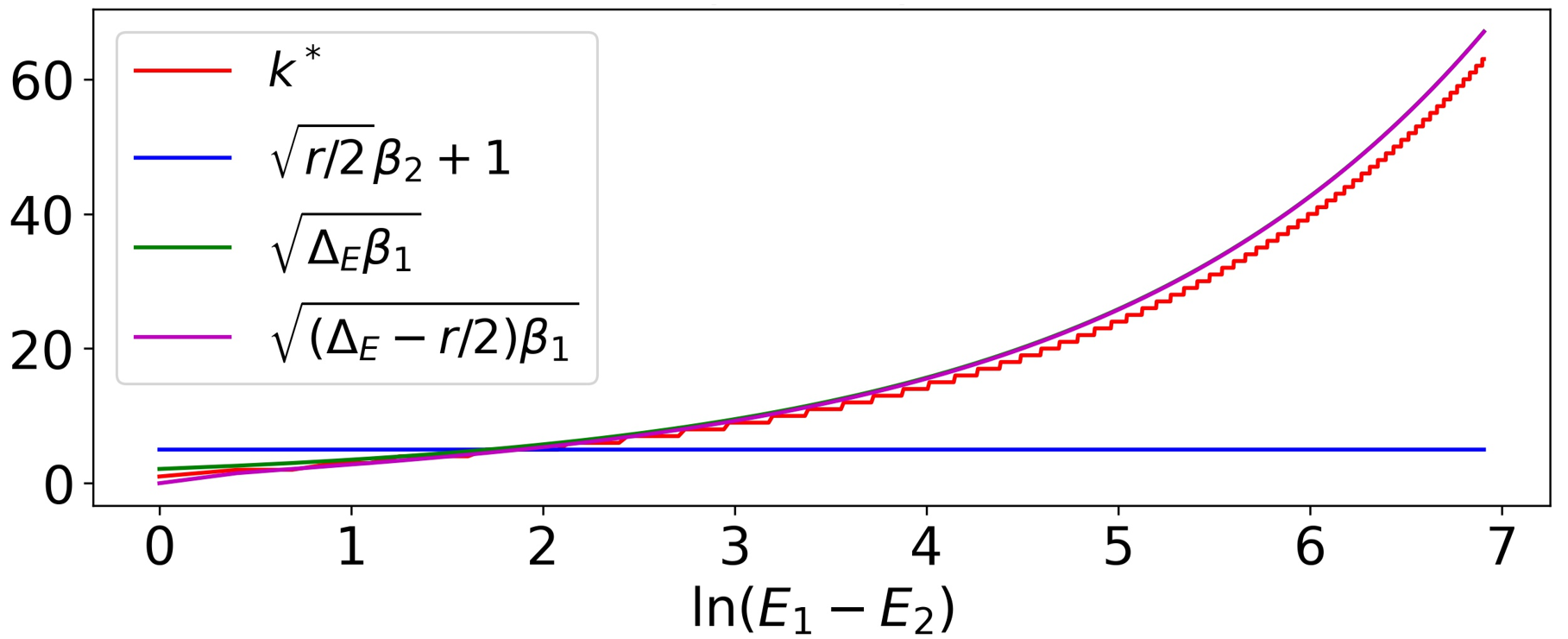}
        \caption{Values of the functions shown in the legend when $\Delta_E \in (1,1000)$.}
        \label{fig:DWFFGcase3-ln}
    \end{subfigure}
    \hfill
    \begin{subfigure}[b]{0.95\linewidth}
        \centering
        \includegraphics[width=\linewidth]{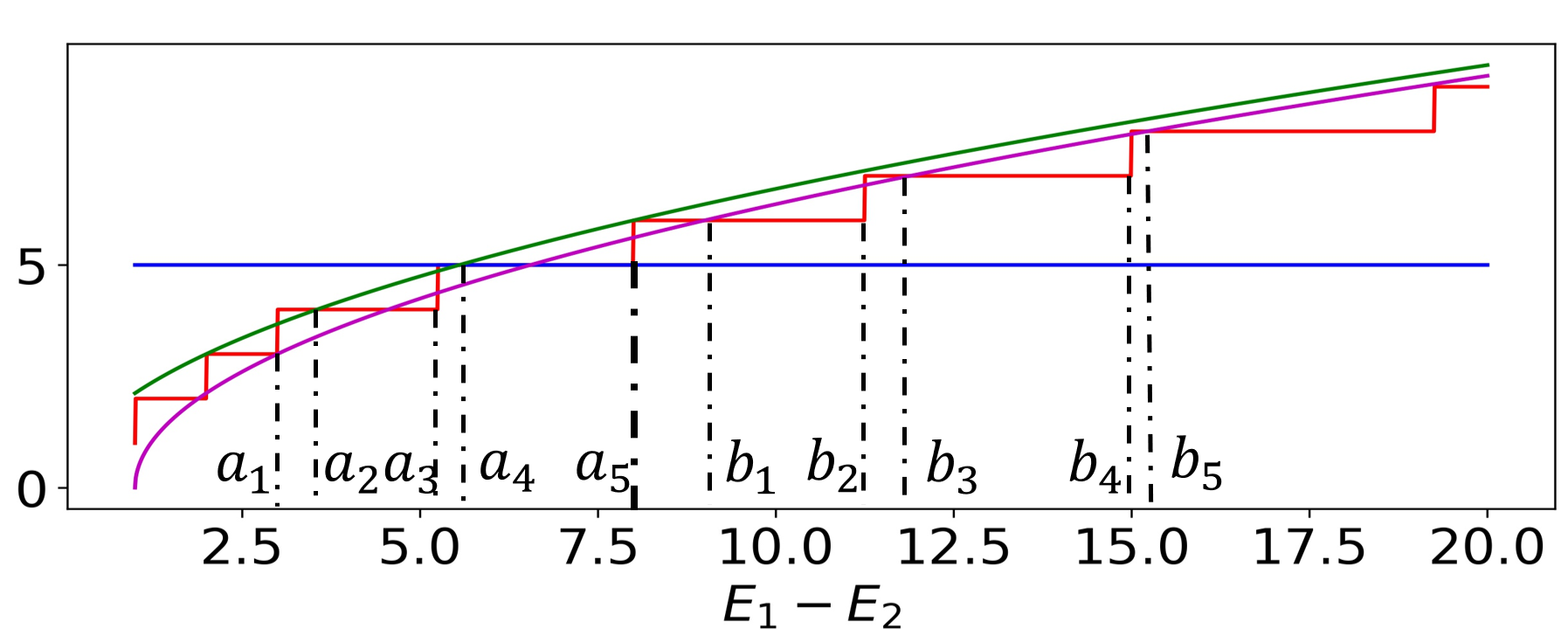}
        \caption{The values of the functions shown in the legend are presented when $\Delta_E$  falls within the range $(1,20)$. It is the beginning region of (a). This depiction serves to elucidate the relationships among these four different functions in detail.}
        \label{fig:DWFFGcase3-detail}
    \end{subfigure}
    \caption{The values of the four functions determine the type of SPEs in Case (3), based on their relations.}
    \label{fig:DWFFGcase3}
\end{figure}

\section{Discussion}\label{sec:discussion}
% \subsection{Comparion with FFG}
In this section, we summarize the differences and similarities among the continuous-time game, the discrete-time game, and the two-agent SFG from \cite{wang2023cooperation} in Table~\ref{table:summaryAllgames}.

\begin{table}[h]
\small
\centering
\caption{Comparison of SPEs in the 2-agent  SFG, the continuous-time game, and the discrete-time game}.
\label{table:summaryAllgames}
\renewcommand{\arraystretch}{1.5} 
\begin{tabular}{|c|c|c|c|} 
\hline
{Game} & SFG {\cite{wang2023cooperation}} & {continuous-time} & {discrete-time}\\
\hline
Strategies & continuous & continuous & discrete \\
\hline
Existence & yes &  no & yes\\
 \hline
Uniqueness & yes & no & no\\
 \hline
$\#$ SPE types & 2 & 3 & 5\\
\hline
A strict flock & possible & never & possible \\
\hline
$t_1^* \leq t_2^*$ & yes & yes & no \\
\hline
\end{tabular}
\end{table}

As mentioned earlier, since agents do not need to arrive at the exact same time to benefit from flocking, the competition is more fierce between agents in FF game than in SFG. In the continuous-time setting, the weaker agent can benefit from flocking while getting the best territory by arriving just ahead of the stronger agent. This also explains the lack of a strict flock in any SPE of the continuous-time game. 

In discrete-time setting, it takes the weaker agent additional travel costs to obtain the better territory. Therefore, if the travel cost is higher than the territorial benefit, then the weaker agent will abandon the competition, arriving no earlier than the stronger agent and in a (strict) flock. Thus the discrete action space effectively leads to less competition compared to the continuous case. 

In general, the more relaxed definition of flocking results in a larger and richer set of SPEs. For instance, in the strict flocking case, if the weaker agent decides to arrive later than the stronger agent, its only option is $t_o$. However, in the present game, it has two options: $t_1+1$ (higher travel cost but lower predation risk) and $t_o$ (lower travel cost but higher predation risk). Another example is the emergence of the possibility of the weaker agent arriving earlier than the stronger agent in the discrete-time case, while in the strict flocking case this cannot occur.

The different types of SPEs in all three sequential games are illustrated in Figure~\ref{fig:3FFG}, where $\Delta_E = E_1 -E_2$ denotes the territorial difference. Figure~\ref{fig:3FFG}(a) shows the three types of SPES in the continuous-time game: (1) when $\Delta_E$ is small, the flock with $t_2^* = t_o$ is the SPE; (2) when $\Delta_E$ is larger, there could be two different SPEs (see Section~\ref{sec:2-agent-continuous} for details): the flock where both arrive earlier or no flocking where Agent 2 arrives at $t_o$. Figure~\ref{fig:3FFG}(b) shows the five types of SPEs in the discrete-time game: (1) when $\Delta_E$ is small enough, the strict flock at $t_o$ is the SPE; (2) when $\Delta_E$ is larger, the flock $(t_o-1,t_o)$ is the SPE; (3) when $\Delta_E$ is even higher, there could be three different SPEs (see Section~\ref{sec:discreteWFFG} for details): a flock $(t_o,t_o-1)$ where the weaker agent arrives first, a flock at earlier times, or no flocking. In this region of the figure (i.e. $\Delta_E > \frac{4}{\beta_2}$), agents on the same line are in the same SPE. By Corollary~\ref{coro:deltaE_large}, when $\Delta_E$ is large enough, there exists a unique SPE, where the stronger agent arrives much earlier ahead while the weaker agent arrives at $t_o$.
Figure~\ref{fig:3FFG}(c) shows the two types of SPEs in SFG: (1) when $\Delta_E$ is small enough, the strict flock at $t_o$ is the SPE; (2) when $\Delta_E$ is larger, there is no flock.

\begin{figure}[!ht]
    \centering
    \begin{subfigure}[b]{0.45\linewidth}
        \centering
        \includegraphics[width=\linewidth]{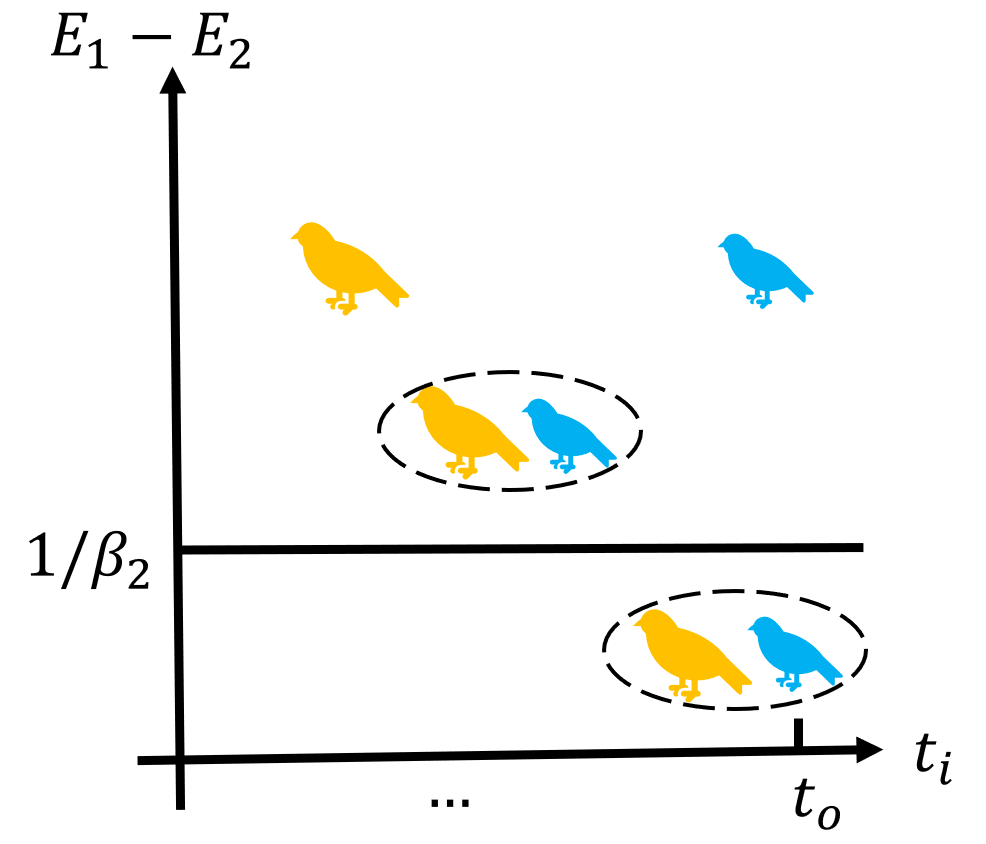}
        \caption{continuous time}
        \label{fig:SPEtypeCT}
    \end{subfigure}
    \hfill
    \begin{subfigure}[b]{0.45\linewidth}
        \centering
        \includegraphics[width=\linewidth]{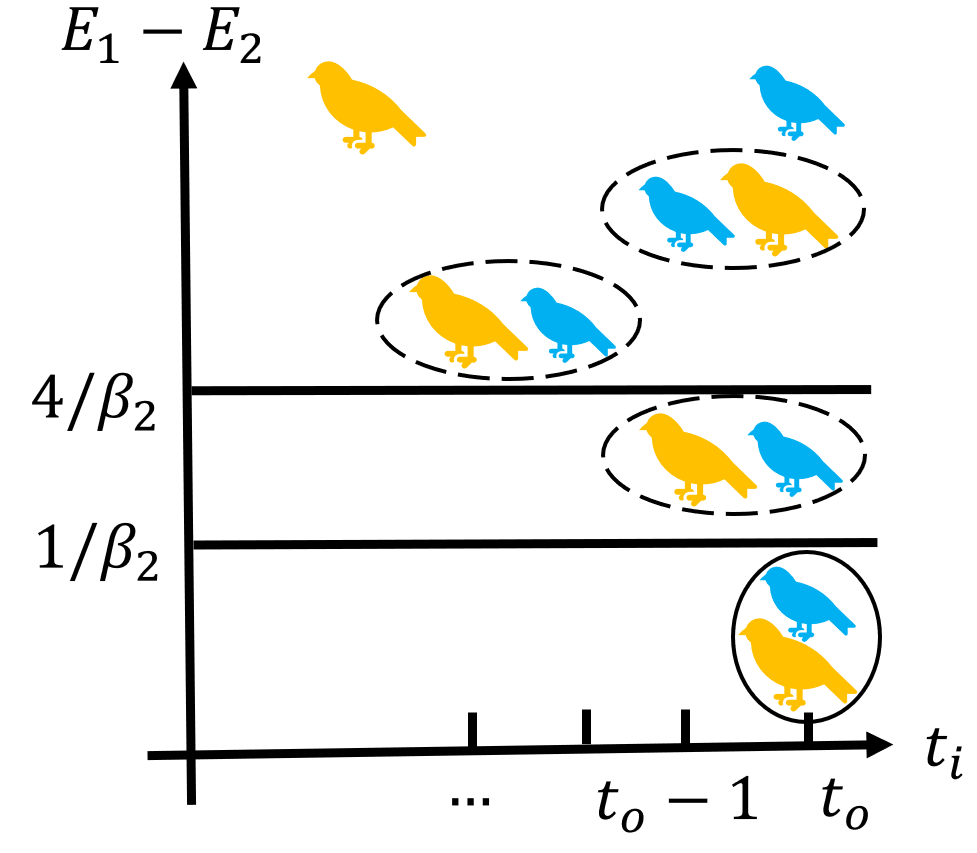}
        \caption{discrete time}
        \label{fig:SPEtypeDT}
    \end{subfigure}
    \hfill
    \begin{subfigure}[b]{0.45\linewidth}
        \centering
        \includegraphics[width=\linewidth]{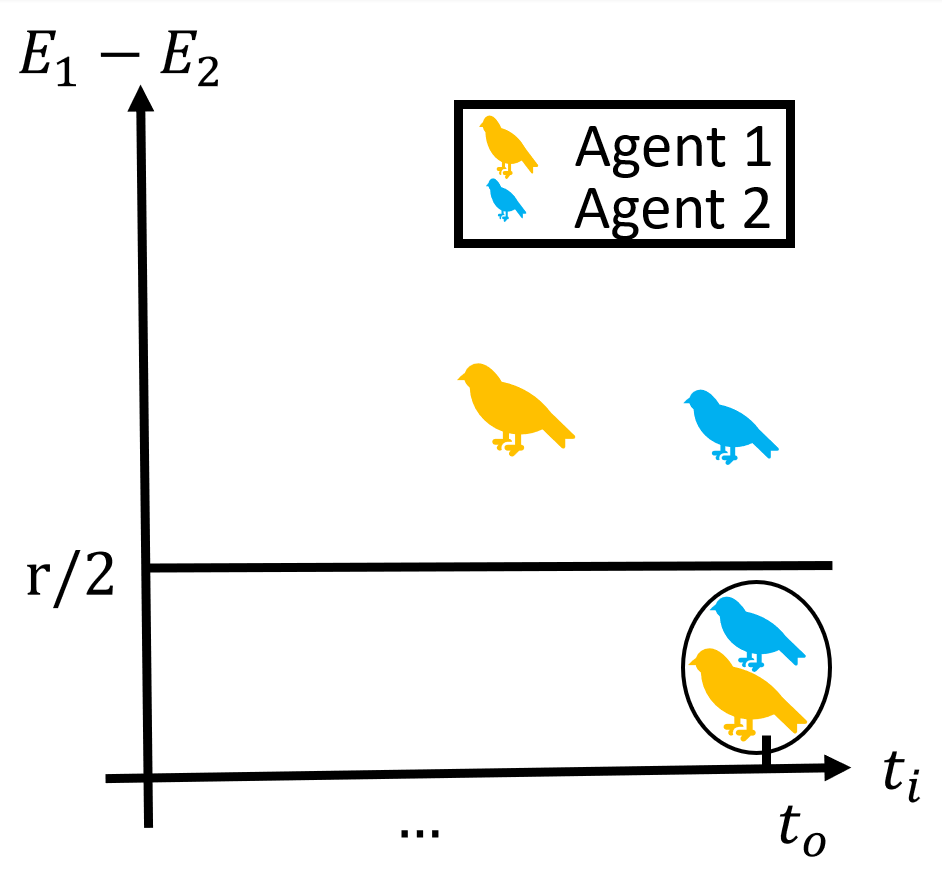}
        \caption{SFG}
        \label{fig:SPEtypeSF}
    \end{subfigure}
    \caption{A simplified illustration of different types of SPEs in continuous-time, discrete-time flock formation game, and SFG. The x-axis is the arrival times of the agents and the y-axis is the difference between the two territories. Depending on the territory differences, there are different cases or games. A solid circle shows a strict flock, while a dashed circle shows a flock but not a strict one.}
    \label{fig:3FFG}
\end{figure}

\section{Conclusion}\label{sec:conclusion} 
 We studied a Stackelberg game on how a flock is formed by individuals arriving within a neighborhood in time. We analyzed the properties of SPEs in both continuous-time and discrete-time settings. Compared with an earlier work requiring flocking to be strictly limited to those arriving at exactly the same time, we show there is a richer set of SPEs in the present game with more intense competition. 

While this study has focused on the 2-agent model, some of the properties and results hold in general for the corresponding $n$-agent game. First, no agent arrives later than $t_o$ in an $n$-agent game, regardless of whether the game is in continuous or discrete time. Secondly, there always exists one or more pure strategy SPEs in an $n$-agent discrete-time game, and developing an algorithm to find a pure strategy SPE for the $n$-agent discrete-time game is a direction of future work. 

\section*{ACKNOWLEDGMENT}
This work has been supported by the ARO under contract W911NF1810208.

%\end{thebibliography}
\nocite{*}
\bibliographystyle{unsrt}
\bibliography{myreference}

\appendix
The following is the full proofs for Theorem~\ref{thm:SPEcon} and Theorem~\ref{thm:SPEdis} respectively.

\begin{proof} [Proof of Theorem~\ref{thm:SPEcon}]
Throughout the proof, we will use $t_{i-}$ to denote the arrival time $t_i-\epsilon$ for small enough $\epsilon>0$. \\
{\bf Case (1): $E_1 - E_2 \leq \frac{1}{\beta_2}$}\\
By Theorem 1, Agent 2 will never strictly cooperate with Agent 1. Agent 2 will cease competing once its arrival time reaches the tipping point, where it becomes indifferent between $t_{1-}$ and the time when Agent 2 obtains the worse territory with the lowest traveling cost (which is $t_o$) while still benefiting from the flocking (i.e. $t_o-t_1 \leq 1$):
$u_2(t_1,t_{1-}) = u_2(t_1, t_o)$ where $u_2(t_1,t_{1-}) = E_1 - \frac{(t_{1-} - t_o)^2}{\beta_2} - \frac{r}{2}$ and $u_2(t_1, t_o) = E_2 - \frac{r}{2}$. Thus, the tipping point is $t_1 = t_o - \sqrt{(E_1-E_2)\beta_2} = t_{1,1}^*$. Given $E_1 - E_2 \leq \frac{1}{\beta_2}$, $t_o-t_{1,1}^* = \sqrt{(E_1-E_2)\beta_2} \leq 1$, indicating that Agent 2 is in the flocking window within one unit after Agent 1's arrival at the tipping point.

If Agent 1 arrives earlier than the tipping point, its utility decreases since the traveling cost increases and the predation risk stays the same or increases if it arrives more than 1 unit earlier from Agent 2. 

If Agent 1 ends up with the worse territory, then the best utility it can obtain is at $t_o$ and receive the benefit from weakly flocking: $u_1(t_o, t_2) = E_2 - \frac{r}{2}$ where $t_2<t_o$ and $t_o-t_2 \leq 1$. If Agent 1 arrives at the tipping point, its utility is $u_1( t_{1,1}^*, t_o) = E_1 - \frac{(E_1-E_2)\beta_2}{\beta_1} - \frac{r}{2}$. Since $\beta_1>\beta_2$, $E_1 - \frac{(E_1-E_2)\beta_2}{\beta_1} > E_2$ and thus $u_1(t_{1,1}^*, t_o) > u_1(t_o, t_2)$. Hence, Agent 1's utility decreases if it ends up with the worse territory.

Therefore, Agent 1's best response is $t_{1,1}^*$ while Agent 2's best response after observing Agent 1's action is $t_{2,1}^* = t_o$.

{\bf Case (2): $E_1 - E_2 > \frac{1}{\beta_2}$}
\begin{enumerate}
    \item If $(E_1-E_2)\beta_2 < \sqrt{2r\beta_2} +1$\\
    Given the larger territory difference in this scenario, leading to increased competition, Agent 2 will withdraw from the competition upon reaching the tipping point. At this point, it becomes indifferent between $t_{1-}$ and the time when Agent 2 secures the inferior territory, while still benefiting from flocking, provided that the travel cost remains affordable compared to the benefits of flocking (i.e. $t_2-t_1 = 1$): $u_2(t_1,t_{1-}) = u_2(t_1, t_1+1)$ where $u_2(t_1,t_{1-}) = E_1 - \frac{(t_{1-} - t_o)^2}{\beta_2} - \frac{r}{2}$ and $u_2(t_1, t_1+1) = E_2 - \frac{(t_{1}+1 - t_o)^2}{\beta_2} -\frac{r}{2}$. Thus, the tipping point is $t_1 = t_o - \frac{(E_1-E_2)\beta_2+1}{2} = t_{1,3}^*$. Notice that $t_{1,3}^* + 1 < t_0$, and that $u_2(t_{1,3}^*,t_{2,3}^*) > u_2(t_{1,3}^*,t_0)$. 
    
    If Agent 1 arrives earlier than the tipping point, its utility decreases since the traveling cost increases, and the predation risk increases as it arrives more than 1 unit earlier than Agent 2. 
    
    If Agent 1 ends up with the worst territory, then the best utility it can obtain is at $t_o$. Since $t_2^*(t_o) = t_{o-}$, Agent 1 receives the benefit from weakly flocking: $u_1(t_o, t_{o-}) = E_2 - \frac{r}{2}$. If Agent 1 arrives at the tipping point, its utility is $u_1(t^*_{1,3}, t^*_{2,3}) = E_1 - \frac{((E_1-E_2)\beta_2+1)^2}{4\beta_1} - \frac{r}{2}$. Depending on whether $u_1(t^*_{1,3}, t^*_{2,3}) \geq u_1(t_o, t_{o-})$ or not, we have the following subcases:
    \begin{enumerate}
        \item If $4(E_1-E_2)\beta_1 \geq ((E_1-E_2)\beta_2+1)^2$, then $u_1(t^*_{1,3}, t^*_{2,3}) \geq u_1(t_o, t_{o-})$, and the unique SPE is $\mathbf{t}^*_3$. 
        \item If $4(E_1-E_2)\beta_1 < ((E_1-E_2)\beta_2+1)^2$, then $u_1(t^*_{1,3}, t^*_{2,3}) < u_1(t_o, t_{o-})$. Therefore, Agent 1 reaches its best utility if it arrives at $t_o$. However, Agent 2 can always increase its utility by arriving slightly sooner than Agent 1, and thus there is no SPE.
    \end{enumerate}
    \item If $(E_1-E_2)\beta_2 > \sqrt{2r\beta_2} +1$
    Given the larger territory difference in this scenario, resulting in increased competition, Agent 2 will cease competition upon reaching the tipping point. At this juncture, it becomes indifferent between $t_{1-}$ and the time when Agent 2 secures the inferior territory with the lowest travel cost (i.e. $t_o$). However, Agent 2 no longer benefits from flocking since the traveling cost to maintain pace within the flocking time frame becomes prohibitively expensive
    (i.e. $t_o-t_1> 1$): $u_2(t_1,t_{1-}) = u_2(t_1, t_o)$ where $u_2(t_1,t_{1-}) = E_1 - \frac{(t_{1-} - t_o)^2}{\beta_2} - \frac{r}{2}$ and $u_2(t_1, t_o) = E_2 - r$. Thus, the tipping point is $t_1 = t_o - \sqrt{\beta_2(E_1 -E_2+\frac{r}{2})} = t^*_{1,2}$. Notice that $t_{1,2}^*  + 1 < t_0$, and that $u_2(t_{1,2}^*,t_{2,2}^*) > u_2(t_{1,2}^*,t_{1,2}^*-1)$. 
    
    If Agent 1 arrives earlier than the tipping point, its utility decreases since the traveling cost increases while the predation risk stays the same.
    
    If Agent 1 ends up with the worst territory, then the best utility it can obtain is at $t_o$. Since $t_2^*(t_o) = t_{o-}$, Agent 1 receives the benefit from weakly flocking: $u_1(t_o, t_{o-}) = E_2 - \frac{r}{2}$. If Agent 1 arrives at the tipping point, its utility is $u_1(t_{1,2}^*,t_{2,2}^*)) = E_1 - \frac{(E_1-E_2+\frac{r}{2})\beta_2}{\beta_1} - r$. Similar to previous case, depending on whether $u_1(t^*_{1,2}, t^*_{2,2}) \geq u_1(t_o, t_{o-})$ or not, we have the following subcases:
    \begin{enumerate}
        \item If $(E_1-E_2-\frac{r}{2})\beta_1 \geq (E_1-E_2+\frac{r}{2})\beta_2$, then $u_1(t^*_{1,2}, t^*_{2,2}) \geq u_1(t_o, t_{o-})$, and the unique SPE is $\mathbf{t}^*_2$. 
        \item If $(E_1-E_2-\frac{r}{2})\beta_1 < (E_1-E_2+\frac{r}{2})\beta_2$, then $u_1(t^*_{1,2}, t^*_{2,2}) < u_1(t_o, t_{o-})$, and there is no SPE. 
    \end{enumerate}
    \item If $(E_1-E_2)\beta_2 = \sqrt{2r\beta_2} +1$
    Given the larger territory difference in this scenario, leading to increased competition, Agent 2 will cease competing once its arrival time reaches the tipping point. At this stage, it becomes indifferent between $t_{1-}$ and the time when Agent 2 secures the inferior territory with the lowest travel cost $t_o$. Also, Agent 2 is indifferent between $t_{1-}$ and $t_1 + 1$ at which it benefits from flocking. That is to say, $u_2(t_1,t_{1-}) = u_2(t_1, t_o) = u_2(t_1, t_1+1)$ where $u_2(t_1,t_{1-}) = E_1 - \frac{(t_{1-} - t_o)^2}{\beta_2} - \frac{r}{2}$, $u_2(t_1, t_o) = E_2 - r$, and $u_2(t_1, t_1 +1) = E_2 - \frac{(t+1+1-t_o)^2}{\beta_2} - \frac{r}{2}$. Thus, the tipping points of Agent 1 is $t_1 = t_o - \sqrt{\beta_2(E_1 -E_2+\frac{r}{2})} = t^*_{1,2}$ and $t_1 = t_o - \frac{(E_1-E_2)\beta_2+1}{2} =t^*_{1,3}$.
    
    The rest of the reasoning is similar to Cases (2) $2)$ and Cases (2) $1)$. Therefore, the SPEs include the ones from Cases (2) $2)$ and Cases (2) $1)$.
\end{enumerate}
\end{proof}

\begin{proof}[Proof of Theorem~\ref{thm:SPEdis}]
    {\bf Case (1): $E_1 - E_2 \leq \frac{1}{\beta_2}$}\\
    Since $E_1 - E_2 \leq \frac{1}{\beta_2}$, then $u_2 (t_1, t_1) = E_2 - \frac{k_1^2}{\beta_2} - \frac{r}{2} \geq u_2(t_1, t_1 - 1) = E_1 - \frac{(k_1+1)^2}{\beta_2} - \frac{r}{2}$ where $t_1 = t_o - k_1$ and $k_1$ is a positive integer. Therefore, the weaker agent strictly prefers to form a strict flock with the stronger agent rather than competing for better territory in a flock. Anticipating $t_2^*(t_1) \geq t_1$, the stronger agent's best response is then $t_1^* = t_o$ as it reaches its global minimum with the lowest traveling cost, the lowest predation cost, and the best territory. Thus, $\mathbf{t}^*_1$ is the unique SPE.

{\bf Case (2): $E_1-E_2 \in (\frac{1}{\beta_2}, \frac{4}{\beta_2}]$}\\
    Since $E_1-E_2 > \frac{1}{\beta_2}$, then $u_2(t_o, t_o) < u_2(t_o, t_o-1)$. Thus, $(t_o, t_o)$ is no longer an SPE. That is to say, Agent 2 wants to compete with Agent 1 for better territory. Therefore, Agent 1 has to arrive at or earlier than $t_o - 1$ otherwise it will end up with a worse territory and a lower utility ($u_1(t_o, t_o) < u_1(t_o-1, t_o-1)$). If $t_1 = t_o -1$, then $u_2(t_1, t_o) = E_2 - \frac{r}{2} \leq u_2(t_1, t_1 - 1) = E_1 - \frac{4}{\beta_2}  - \frac{r}{2}$ and thus Agent 2 is indifferent between arriving one unit time earlier than arriving later in the flock or strictly prefers the latter (i.e. $t_2^* (t_1 = t_o-1) = t_o$).  Therefore, Agent 1's best response, anticipating Agent 2's decision, is $t_1^* = t_o -1$. As a result,  $\mathbf{t}^*_2$ is the unique SPE.
    % $(t_o-1, t_o)$ 

{\bf Case (3): $E_1-E_2>\frac{4}{\beta_2}$}\\ 
Since $E_1-E_2 > \frac{4}{\beta_2}$, it follows that $u_2(t_o-1, t_o) < u_2(t_o-1, t_o-2)$. Consequently, $(t_o-1, t_o)$ is no longer an SPE. This implies that Agent 2 aims to compete with Agent 1 for the better territory even if Agent 1 arrives at $t_o - 1$. Therefore, Agent 1 must arrive at or before $t_o - 2$ (i.e., $t_1 \leq t_o-2$), or else it risks ending up with a worse territory and lower utility. If Agent 1 arrives at $t_1 \leq t_o -2$, Agent 2's optimal response is limited to three choices: $t_1 - 1$, $t_1 +1$, and $t_o$. In deriving the arrival time $t_1^* = t_o - k_1$ for some integer $k_1\geq 2$  under an SPE, we consider
two possibilities: either Agent 1 secures the best territory $E_1$ in the SPE, or it does not. 

If Agent 1 secures the best territory $E_1$ in the SPE, then it means that Agent 1 arrives ahead of Agent 2 or at the same time as Agent 2 (clearly this never happens here since $t_2^*(t_1)\in \{t_1-1, t_1+1, t_o\}$ and $t_1 + 1< t_o$). When Agent 1 arrives at $t_o - k_1$ and obtains the best territory, the best response of Agent 2 is either $t_o - k_1+1$ or $t_o$ (i.e. $t_o - k_1-1$ is no better than at least one of the other two choices) given: 
\begin{equation}
    u_2(t_o - k_1, t_o - k_1-1) \leq u_2(t_o - k_1, t_o - k_1+1)
\end{equation}
or 
\begin{equation}
    u_2(t_o - k_1, t_o - k_1-1) \leq u_2(t_o - k_1, t_o)
\end{equation}
This also implies that Agent 1 can't secure the best territory if it arrives one unit time later at $t_o - k_1 + 1$, which means that Agent 2 is willing to arrive earlier to get the best territory. Therefore, if $t_1 = t_o - k_1 + 1$, then $t_o - k_1$ is the best response for Agent 2, and thus both $t_o - k_1+2$ and $t_o$ are strictly worse than $t_o - k_1$:
\begin{equation}
    u_2(t_o - k_1+1, t_o - k_1) > u_2(t_o - k_1+1, t_o - k_1+2)
\end{equation}
\begin{equation}
    u_2(t_o - k_1+1, t_o - k_1) > u_2(t_o - k_1+1, t_o)
\end{equation}
Therefore, we can find the arrival of Agent 1 $t_o-k_1$ where $k_1 = \lceil\min(\sqrt{(E_1-E_2+\frac{r}{2})\beta_2}, \frac{(E_1-E_2)\beta_2}{4}+1) \rceil - 1$.\\
When Agent 1 arrives at $t_o-k_1$, Agent 2 strictly prefers $t_o-k_1+1$ over $t_o$ if $u_2(t_o-k_1, t_o-k_1+1) > u_2(t_o-k_1, t_o)$ (i.e. $k_1 < \sqrt \frac{r}{2}\beta_2+1$). Therefore, if $k_1 > \sqrt \frac{r}{2}\beta_2+1$, Agent 2's best response is $t_o$; if $k_1 = \sqrt \frac{r}{2}\beta_2+1$, Agent 2's best responses are $t_o$ and $t_o-k_1+1$.

If Agent 1 does not secure the best territory $E_1$ and ends up with the worst territory $E_2$, then its best response is to arrive at $t_o$ with the lowest travel cost. Observing $t_1 = t_o$, Agent 2's best response is $t_o -1$ since $u_2(t_o, t_o-1) > u_2(t_o, t_o)$. Therefore, $u_1(t_o, t_o-1) = E_2 - \frac{r}{2}$. 

Depending on the sets of conditions, we have only three possible SPEs: $\mathbf{t}^*_3$, $\mathbf{t}^*_4$, and $\mathbf{t}^*_5$.
\begin{enumerate}
    \item If $k_1 \leq \sqrt \frac{r}{2}\beta_2+1$, the only possible SPEs are:  $(t_o-k_1, t_o-k_1+1)$ and $(t_o, t_o-1)$. This includes two different cases: if $k_1 < \sqrt \frac{r}{2}\beta_2+1$, then Agent 2 strictly prefers $t_o-k_1+1$ over $t_o$ and thus the only possible SPEs are:  $(t_o-k_1, t_o-k_1+1)$ and $(t_o, t_o-1)$; If $k_1 = \sqrt \frac{r}{2}\beta_2+1$, then Agent 2 is indifferent between $t_o-k_1+1$ and $t_o$ and thus all three are possible SPEs: $(t_o-k_1, t_o-k_1+1)$, $(t_o-k_1, t_o)$ and $(t_o, t_o-1)$. Since $u_1(t_o-k_1, t_o-k_1+1) > u_1(t_o-k_1, t_o)$, Agent 1 strictly prefers $t_o-k_1$ over $t_o$. Therefore, only $\mathbf{t}^*_3$ and $\mathbf{t}^*_5$ are the possible SPEs. Thus, we can check the two cases together.
    \begin{enumerate}
        \item If $k_1<\sqrt{(E_1 - E_2)\beta_1}$, then $u_1(t_o-k_1, t_o-k_1+1) > u_1(t_o, t_o-1)$ and thus
        $\mathbf{t}^*_3$ is {the unique} SPE.
        \item If $k_1>\sqrt{(E_1 - E_2)\beta_1}$, then $u_1(t_o-k_1, t_o-k_1+1) < u_1(t_o, t_o-1)$ and thus 
        $\mathbf{t}^*_5$ is {the unique} SPE.
        \item If $k_1=\sqrt{(E_1 - E_2)\beta_1}$, then $u_1(t_o-k_1, t_o-k_1+1) = u_1(t_o, t_o-1)$ and thus 
        $\mathbf{t}^*_3$ and $\mathbf{t}^*_5$ are the two SPEs.
    \end{enumerate}
    \item If $k_1 > \sqrt \frac{r}{2}\beta_2+1$, then Agent 2 strictly prefers $t_o$ over $t_o-k_1+1$ and thus the only possible SPEs are:  $\mathbf{t}^*_4$ and $\mathbf{t}^*_5$.
    \begin{enumerate}
        \item If $k_1<\sqrt{(E_1 - E_2-\frac{r}{2})\beta_1}$, then $u_1(t_o-k_1, t_o) > u_1(t_o, t_o-1)$ and thus
        $\mathbf{t}^*_4$ is {the unique} SPE.
        \item If $k_1>\sqrt{(E_1 - E_2-\frac{r}{2})\beta_1}$, then $u_1(t_o-k^*, t_o) < u_1(t_o, t_o-1)$ and thus $\mathbf{t}^*_5$ is {the unique} SPE.
        \item If $k_1=\sqrt{(E_1 - E_2-\frac{r}{2})\beta_1}$, then $u_1(t_o-k_1, t_o) = u_1(t_o, t_o-1)$ and thus
        $\mathbf{t}^*_4$ and $\mathbf{t}^*_5$ are the two SPEs.
    \end{enumerate}
\end{enumerate}
\end{proof}

\end{document}